%% file: main.tex
\documentclass{article}

\usepackage[a4paper, width=15cm, height=220mm]{geometry}

\usepackage[mathscr]{euscript}
\usepackage{graphicx}
\usepackage{subfigure} 
\usepackage{algorithm}
\usepackage{algorithmic}
\usepackage{amsfonts}
\usepackage{amsthm}
\usepackage{amsmath}
\usepackage[utf8]{inputenc}   

\newtheorem{theorem}{Theorem}
\newtheorem{definition}{Definition}
\newtheorem{lemma}{Lemma}
\newtheorem*{proposition*}{Proposition}
\newtheorem{corollary}{Corollary}
\newtheorem{example}{Example}
\newtheorem{remark}{Remark}
\newtheorem{proposition}{Proposition}
\newtheorem*{construction*}{Construction}

\newcommand{\vs}[1]{\ensuremath{\mathbf{\boldsymbol{#1}}}}
\renewcommand{\v}[1]{\ensuremath{\mathbf{#1}}}
\newcommand{\e}[1]{\ensuremath{\boldsymbol{\mathscr{#1}}}}

\newcommand{\R}{\mathbb{R}}

\newcommand{\N}{\mathbb{N}}

\newcommand{\K}{\mathbb{K}}
\newcommand{\C}{\mathbb{C}}
\newcommand{\comment}[1]{}

\newcommand{\Ac}{\mathcal{A}}
\newcommand{\Bc}{\mathcal{B}}
\newcommand{\Cc}{\mathcal{C}}
\newcommand{\Mc}{\mathcal{M}}

\title{Recognizable Series on Hypergraphs%\thanks{we thank...}
}

\author{Rapha\"el BAILLY\and Fran\c{c}ois DENIS\and Guillaume RABUSSEAU\thanks{ \texttt{rbailly@lsi.upc.edu},  \texttt{\{francois.denis, guillaume.rabusseau\}@lif.univ-mrs.fr}}}

\begin{document}
\maketitle
\begin{abstract}
\input{abstract.tex}
\end{abstract}

\section{Introduction}
	\input{introduction.tex}
\vspace{-0.15cm}
\section{Preliminaries}\label{sec:prel}

	\subsection{Recognizable Series on Strings and Trees}
		\input{prel_rational.tex}
	\subsection{Tensors}
		\input{prel_tensors.tex}
	\subsection{Hypergraphs}

\input{prel_hypergraphs.tex}

\section{Hypergraph Weighted Models}\label{sec:hwm}

	\subsection{Definition}
		\input{hwm_definition.tex}
		\subsection{Properties}
			\input{hwm_string_trees.tex}

\input{hwm_addition.tex}

			\input{hwm_product.tex}
\section{Recognizability of Finite Support Series}\label{sec:tilings}

\input{hwm_tilings_motivation.tex}

	\subsection{Tilings}
			\input{hwm_finiteSupport.tex}

	\section{Examples}
		\input{hwm_examples.tex}

\section{Conclusion}
	\input{conclusion.tex}

\bibliographystyle{plain}
\bibliography{mfcs}
\nocite{*}
\newpage
%\section{Appendix}
%\input{appendix.tex}
\end{document}

%% file: abstract.tex
We introduce the notion of \emph{Hypergraph Weighted
  Model} (HWM) that generically associates a tensor network to a hypergraph and then computes a value by tensor contractions directed by its hyperedges. A series $r$ defined on a hypergraph family is said to be recognizable if there exists a HWM that computes it. This model generalizes the notion of rational series on strings and trees. We prove some  properties of the model and study at which conditions finite support series are recognizable.

%%% Local Variables: 
%%% mode: latex
%%% TeX-master: "main"
%%% End: 

%% file: introduction.tex
Real-valued functions whose domains are composed of syntactical
structures, such as strings, trees or graphs, are widely used in
computer science. One way to handle them is by means of rational
series that use automata devices to jointly analyze the structure of the input and
compute its image. Rational series have been defined for strings and
trees, but their extension to graphs is challenging.

On the other hand, rational series have an equivalent algebraic
characterization by means of linear (or multi-linear)
representations.  We show in this paper that this last formalism can
be naturally extended to graphs (and hypergraphs) by associating tensors
to the vertices of the graph.  

More precisely, we define the notion of \emph{Hypergraph Weighted
  Model} (HWM), a computational model that generically associates a tensor
network to a hypergraph and that computes a value by
successive generalized tensor contractions directed by its
hyperedges. We say that a series $r$ defined on a hypergraph family is
HWM-recognizable if there exists a HWM $M$ that computes it: we
then denote $r$ by $r_M$. We first show
that HWM-recognizable series defined on strings or trees exactly
recover the classical notion of recognizable series. We present two closure
properties: if $r$ and $s$ are two recognizable series defined on
a family $\mathcal{H}$ of connected hypergraphs, then $r+s$ and $r\cdot s$, respectively
defined for all graph $G \in \mathcal{H}$ by $(r+s)(G)=r(G)+s(G)$ and $(r\cdot s)(G)=r(G)s(G)$ (the
Hadamard product), are HWM-recognizable.

Recognizable series on strings and trees include polynomials, i.e. finite
support series. This is not always the case for recognizable
series defined on more general families of hypergraphs. For example, we show that finite support series are not
recognizable on the family of circular strings. The main
reason is that if a recognizable series is not null on
some hypergraph $G$, it must be also different from zero on
\emph{tilings} of $G$, i.e. connected graphs made of copies of $G$. We
show that if a graph family is tiling-free, then finite support series are recognizable. Strings and trees, as any family of
rooted hypergraphs, are tiling-free.

String rational series and weighted automaton have their roots in automata theory \cite{Eilenberg_1974,Schutzenberger_1961} and their study can be found in \cite{berstel1988rational,Droste_Kuich_Vogler_2009,Kuich_Salomaa_1986,sakarovitch2009,Salomaa_Soittola_Bauer_Gries_1978}. The extension of rational series and weighted automaton to trees is presented in 
\cite{Berstel82,Droste_Kuich_Vogler_2009}. Spectral methods for inference of stochastic languages of strings/trees have been developed upon the notion of linear representation of a rational series (\cite{Bailly:2009:GIP:1553374.1553379,Denis07} for example). Tensor networks emerged in the theory of brain functions \cite{Pellionisz_Llinas_1979}, they have been used in quantum  theory (see for example \cite{Orus:arXiv1306.2164}), and the interest for these objects has recently been growing in other fields (e.g. data mining \cite{ABSP0001sub2}).

We recall notions on tensors and hypergraphs in Section~\ref{sec:prel},
we introduce the Hypergraph Weighted
  Model and present some of its properties
  in section~\ref{sec:hwm}, we introduce the notion of tilings and we
  study the recognizability of finite support series in
  Section~\ref{sec:tilings}, we provide some examples in Section~\ref{sec_examples} and we then propose a short conclusion.

    Most of the proofs have been omitted for brevity  but can be found in \cite{bailly14}.

%% file: prel_rational.tex
We refer to \cite{Berstel82,berstel1988rational,tata,Droste_Kuich_Vogler_2009,sakarovitch2009} for notions about recognizable series on strings and trees, and we briefly recall below some basic definitions.

Let $\Sigma$ be a finite \emph{alphabet}, and $\Sigma^{*}$ be the set of strings on $\Sigma$. A \emph{series} on $\Sigma^*$ is a mapping $r: \Sigma^*\rightarrow \K$ where $\K$ is a semiring.  A series $r$ is \emph{recognizable}  if there exists a tuple $\langle V,\vs{\iota}, \{\v{M}_x\}_{x\in \Sigma}, \vs{\tau}\rangle$ where $V=\K^d$ for some integer $d\geq 1$, $\vs{\iota}, \vs{\tau} \in V$ and $\v{M}_x\in \K^{d\times d}$ for each symbol  $x\in \Sigma$ , such that for any $u_1\dots u_n\in\Sigma^*$, $r(u_1\dots u_n)=\vs{\iota}^{\top}\v{M}_{u_1}\dots \v{M}_{u_n}\vs{\tau}$. In this paper, we will only consider the case where $\K = \R$ or $\C$.

A \emph{ranked} alphabet ${\cal F}$ is a tuple $(\Sigma,\sharp)$ where $\Sigma$ is a finite alphabet and where $\sharp$ maps each symbol $x$ of $\Sigma$ to an integer $\sharp x$ called its \emph{arity}; for any $k\in \N$, let us denote ${\cal F}_k = \sharp^{-1}(\{k\})$. A ranked alphabet is \emph{positive} if $\sharp$ takes its values in $\N_+$.

The set of trees over a ranked alphabet ${\cal F}$ is denoted by $T({\cal F})$. A \emph{tree series} on $T({\cal F})$ is
  a mapping $r:T({\cal F}) \to \K.$ A series $r$ is \emph{recognizable} 
  if there exists a tuple $\langle V,\mu,\vs{\lambda} \rangle$,
  where $V=\K^d$ for some integer $d\geq 1$,  $\mu$ maps each  $f\in {\cal F}_p$ to a $p$-multilinear mapping $\mu(f)\in\mathcal{L}(V^p;V)$ for each $p \geq 0$  and $\vs{\lambda}\in V$, such that $r(t) = \vs{\lambda}^{\top}\mu(t)$ for all $t$ in $T(\mathcal{F})$, where $\mu(t)\in V$ is inductively defined by $\mu(f(t_1, \dots, t_p) ) = \mu(f)(\mu(t_1), \dots, \mu(t_p))$.

%%% Local Variables: 
%%% mode: latex
%%% TeX-master: "main"
%%% End: 

%% file: prel_tensors.tex
Let $d\geq 1$ be an integer, $V=\K^d$ where $\K=\R$ or $\C$ and let $(\v{e}_{1}, \dots, \v{e}_{d})$ be the canonical basis of $V$. A tensor $\e{T}\in \bigotimes^k V = V \otimes \cdots \otimes V$ ($k$ times) can uniquely be expressed as a linear combination 
$$ \e{T}=\sum_{i_1, \dots ,i_k \in [d]} \e{T}_{i_1 \dots i_{k}} \v{e}_{i_1} \otimes \dots  \otimes \v{e}_{i_k}$$
(where $[d] = \{1, \cdots, d\}$) of \emph{pure tensors} $\v{e}_{i_1} \otimes \dots  \otimes \v{e}_{i_k}$ which form a basis of $\bigotimes^k V$ \cite{bookTensor}. Hence, the tensor $\e{T}$ can be represented as the multi-array $(\e{T}_{i_1 \dots i_{k}})$.

\begin{definition}
The \emph{tensor product} of $\e{T}\in \bigotimes^{p} V$ and $\e{U}\in \bigotimes^{q} V$ is the tensor $\e{T}\otimes \e{U}\in \bigotimes^{p+q} V$ defined by $$(\e{T}\otimes \e{U})_{i_1 \cdots i_p j_1 \cdots j_q}=\e{T}_{i_1 \cdots i_p} \e{U}_{j_1 \cdots j_q}.$$
\end{definition}
For any $\v{v}\in\K^d$, let $\v{v}^{\otimes k} = \v{v}\otimes \cdots \otimes \v{v}=\sum_{i_1, \dots ,i_k \in [d]} v_{i_1}\dots v_{i_k}\v{e}_{i_1} \otimes \dots  \otimes \v{e}_{i_k}$ denote its $k$-th tensor power.

Let $\odot: V \times V \to V$ be an associative and symmetric bilinear mapping: $\forall u,v,w \in V, u \odot v=v\odot u$ and $ u \odot ( v \odot w)= (u \odot v) \odot w$. The mapping $\odot$ is called a \emph{product}.
\begin{remark}
\label{rmk_contraction}
Let $\mathbf{1}=(1, \dots, 1)^{\top}$ and let $\odot_{id}$ be defined by $\v{e}_i \odot_{id} \v{e}_j= \delta_{ij}\v{e}_i$, where $\delta$ is the Kronecker symbol: $\odot_{id}$ is called the \emph{identity product}. 

The operation of applying the linear form $\v{v} \mapsto \mathbf{1}^\top \v{v}$ to the identity product  $\v{a} \odot_{id} \v{b}$ of two vectors is related to the notions of \emph{generalized trace} and \emph{contraction}: if $\e{A} = \sum_{i,j \in [d]} \e{A}_{i,j} \v{e}_i \otimes \v{e}_j$ is a 2-order tensor over $\K^d$ (i.e. a square matrix), $\v{v} = \sum_{i,j\in [d]} \e{A}_{i,j} \v{e}_i \odot_{id} \v{e}_j$ is the diagonal  vector of $\e{A}$ and $ \mathbf{1}^\top \v{v}$ is its trace. Furthermore, if  $\e{A} = \sum_{i,j \in [d]} \e{A}_{i,j} \v{e}_i \otimes \v{e}_j$ and $\e{B} = \sum_{i,j \in [d]} \e{B}_{i,j} \v{e}_i \otimes \v{e}_j$ are $2$-order tensors over $\K^d$, then $  \sum_{i,j,k,l} \e{A}_{i,j} \e{B}_{k,l} \v{e}_i \otimes \mathbf{1}^\top (\v{e}_j \odot_{id} \v{e}_k) \otimes \v{e}_l = \sum_{i,j,l} \e{A}_{i,j} \e{B}_{j,l} \v{e}_i \otimes \v{e}_l$ is the tensor form of the matrix product $\e{A} \cdot \e{B}$ (i.e. the contraction of the tensor $\e{A}\otimes \e{B}$ along its 2nd and 3rd modes).
\end{remark}

%%% Local Variables: 
%%% mode: latex
%%% TeX-master: "main"
%%% End: 

%% file: prel_hypergraphs.tex
\begin{definition}
A hypergraph $G=(V,E,l)$ over a positive ranked alphabet $(\Sigma,\sharp)$ is given by a non empty finite set $V$, a mapping $l: V\to \Sigma$ and a partition $E=(h_k)_{1\leq k\leq n_E}$ of $P_G =\{(v,j): v\in V, 1\leq j\leq \sharp v\}$ where $\sharp v = \sharp l(v)$. 
\end{definition}

$V$ is the set of \emph{vertices}, $P_G$ is the set of \emph{ports} and $E$ is the
set of \emph{hyperedges} of $G$. The arity of a symbol $x$ is equal to to the number of ports of any vertex labelled by $x$. 
We will sometimes use the notation $v^{(i)}$ for
the port $(v,i) \in P_G$. A hypergraph $G$ can be represented as a bipartite graph where vertices from one partite set represent
the vertices of $G$, and vertices from the other represent its hyperedges (see Figure~\ref{ex-1}).  
%A hypergraph is \emph{closed} if every hyperedgecontains at least two ports (every port of every vertex is connected to another one). 
A hypergraph is \emph{connected} if for any partition
$V=V_1\cup V_2$, there exists a hyperedge $h\in E$ and
ports $v_1^{(i)},v_2^{(j)}\in h$ s.t. $v_1\in V_1$ and $v_2\in V_2$. A hypergraph is a \emph{graph} if $|h| \leq 2$ for all $h\in E$, and a hypergraph is \emph{closed} if $|h| \geq 2$ for all $h\in E$. 
%In the description of hypergraphs, the single edges will be omitted.

\begin{example}\label{ex-1}
Over the ranked alphabet $\{(a,3),(b,2)\}$, let $V=\{v_1,v_2,v_3\}$, $l(v_1)=l(v_3)=a$, $l(v_2)=b$, $E=\{h_1, h_2, h_3, h_4\}$ where $h_1=\{v_1^{(1)}, v_3^{(3)}\}$, $h_2=\{v_1^{(2)},v_2^{(1)},v_3^{(2)}\}$, $h_3=\{v_1^{(3)},v_2^{(2)}\}$ and $h_4=\{v_3^{(1)}\}$ (see Figure~\ref{ex-1}).   
% Let us consider the following hypergraph $G_0$, with $V=\{ (a,\{a_1^1,a_1^2,a_1^3\}), (b,\{b_2^1,b_2^2\}),(a,\{a_3^1,a_3^2,a_3^3\})\}$ and $E=\{ \{ a_1^1,a_3^3\} , \{ a_1^2,b_2^1,a_3^2\}, \{ a_1^3,b_2^2\}, \{ a_3^1\}\}$. It is represented Fig.\ref{fig-1}. $E$ can also be written $E=\{ \{ a_1^1,a_3^3\} , \{ a_1^2,b_2^1,a_3^2\}, \{ a_1^3,b_2^2\}\}$, the remaining ports (here $a_3^1$) are considered as single edges.
\begin{figure}
\begin{center}

\vspace{-0.5cm}
\includegraphics[scale=1]{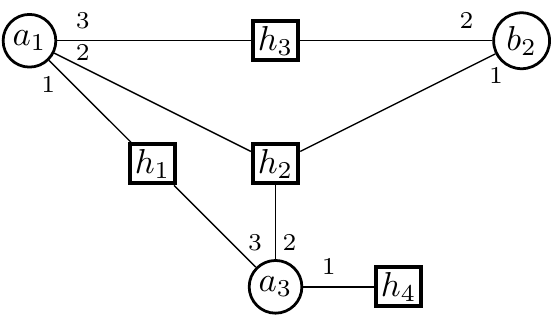}
\end{center}
\vspace{-0.25cm}
\label{fig-1}\caption{The hypergraph $G$ from example~\ref{ex-1}.}
\end{figure}

\end{example}

\begin{example}\label{def_hypergraphs_stringEx}A string $u=u_1\dots u_n$ over an alphabet $\Sigma$ can be seen as a (hyper)graph over the ranked alphabet $(\Sigma \cup \{\iota, \tau\}, \sharp)$ where $\sharp x=2$ for any $x\in \Sigma$ and $\sharp \iota = \sharp \tau = 1$. Let $V=\{0,\cdots, n+1\}$, $l(0) = \iota$, $l(n+1) = \tau$ and $l(i)=u_{i}$ for $1\leq i\leq n$. Let $E=\{h_0, h_1, \dots, h_{n}\}$ where $h_0=\{(0,1),(1,1)\}$ and $h_i=\{(i,2),(i+1,1)\}$ for $1\leq i\leq n$ (see Figure~\ref{fig_string}). The set of strings $\Sigma^*$ gives rise to a family of hypergraphs. 
\begin{figure}
\label{fig_string}
\begin{center}

\includegraphics[scale=1]{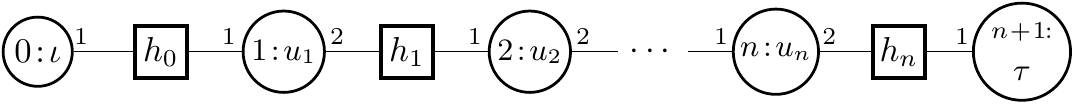}
\end{center}
\vspace{-0.25cm}
\caption{Graph associated with a string $u = u_1\cdots u_n$ (where the notation $i:x$ means that $\ell(i)=x$)}
\end{figure}

\end{example}

\begin{example}\label{def_hypergraphs_treeEx}
Similarly, we can associate any tree $t$ over a ranked alphabet $(\Sigma, \sharp)$ with a graph $G_t$ on the ranked alphabet $(\Sigma \cup \{\lambda\}, \sharp')$ where $\sharp '(f) = \sharp f + 1$ for any $f \in \Sigma$, and where the special symbol $\lambda$ of arity 1 is connected to the free port of the vertex corresponding to the root of $t$. 

Formally, let ${\cal F}=(\Sigma,\sharp)$ be a ranked alphabet. A tree $t$ over ${\cal F}$ can be defined as a mapping  from a finite non-empty prefix-closed set $Pos(t)\subseteq \N^*$ to ${\cal F}$, satisfying the following conditions: (i) $\forall p \in Pos(t)$, if $t(p) \in {\cal F}_n,n \geq 1$, then $\{j | p\cdot j\in Pos(t)\} = \{1,...,n\}$, (ii) $\forall p\in Pos(t)$, if $t(p)\in {\cal F}_0$, then $\{j|p\cdot j\in Pos(t)\}=\emptyset$.

A tree $t$ over ${\cal F}$  can be seen as a hypergraph over the ranked alphabet $(\Sigma \cup \{\lambda\}, \sharp ')$ where $\sharp '(\lambda) = 1$ and $\sharp '(f)=\sharp f +1$ for any $f\in \Sigma$. Let $V=Pos(t) \cup \{0\}$, $l(0) = \lambda$ and $l(p)=t(p)$ for any $p\in Pos(t)$. Let $E=\left\{\{(0,1),(\varepsilon,1)\}\right\}\cup \bigcup_{p.j\in Pos(t)} \left\{\{(p,j+1),(p.j,1)\}\right\}$. The set of trees $T({\cal F})$ gives rise to a family of hypergraphs. The graph associated with the tree $t=f(a,f(a,a))$ is shown as an example in Figure~\ref{fig_circ_trees}.
\comment{
For example, let ${\cal F}=\{f(\cdot,\cdot),a\}$ and let $t=f(a,f(a,a))$:\
\begin{itemize}
\item $V=\{0,\varepsilon, 1, 2, 2.1, 2.2\}$,
\item $l(0) = \lambda$,
  $l(\varepsilon)=l(2)=f$ and $l(1)=l(2.1)=l(2.2)=a$
\item 
  $E=\{\{(0,1),(\varepsilon,1)\},\{(\varepsilon,2),(1,1)\},\{(\varepsilon,3),(2,1)\},\{(2,2),(2.1,1)\},\{(2,3),(2.2,1)\}\}.$
\end{itemize}
}
\end{example}
\begin{example}
\label{ex_circular_strings}
Given a finite alphabet $\Sigma$, let ${\cal F} = (\Sigma, \sharp)$ be the ranked alphabet where $\sharp x = 2$ for each $x\in \Sigma$. We say that a hypergraph $G = (V,E)$ on $\mathcal{F}$ is a circular string if and only if  $G$ is connected and every hyperedge $h\in E$ is of the form $h = \{(v, 2), (w,1)\}$ for $v,w \in V$ (see Figure~\ref{fig_circ_trees}). 
\end{example}

\begin{example}
An other interesting extension of strings (naturally modeled by graphs) is the set of 2D-words $w\in \Sigma^{M\times N}$ on a finite alphabet $\Sigma$, see Section~\ref{sec_crosswords} for details.
\end{example}

\begin{figure}
\vspace{-0.5cm}
\centering
\begin{minipage}{.5\textwidth}
 \begin{center}
\includegraphics[scale=1]{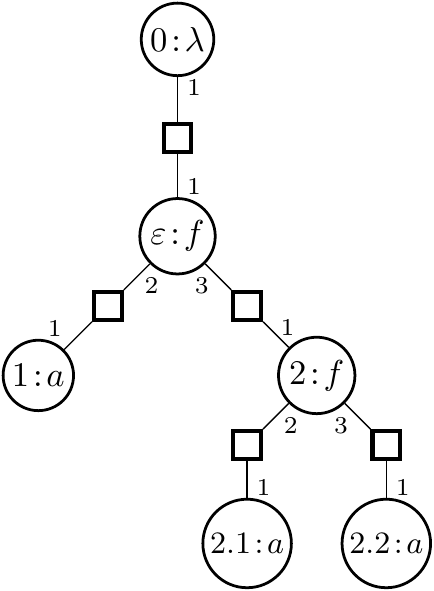}
\end{center} 
\end{minipage}%
\begin{minipage}{.5\textwidth}
\begin{center}
\includegraphics[scale=1]{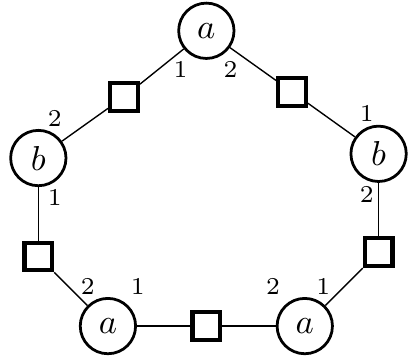}
\end{center} 
\end{minipage}
\caption{(left) Hypergraph $G_t$ associated with the tree $t = f(a,f(a,a))$. (right) Example of circular string on the alphabet $\{a,b\}$}
\label{fig_circ_trees}
\end{figure}

%%% Local Variables: 
%%% mode: latex
%%% TeX-master: "main"
%%% End: 

%% file: hwm_definition.tex
In this section, we give the formal definition of Hypergraph Weighted Models. We then explain how to compute its value for a given hypergraph.

\begin{definition}
\label{def_hwm}
A rank $d$ Hypergraph Weighted Model (HWM) on a ranked alphabet $(\Sigma,\sharp)$ is a tuple $M=\langle V_M, \{\e{T}^x\}_{x\in \Sigma}, \odot, \vs{\alpha} \rangle$ where $V_M=\K^d$, $\odot$ is a product on $V_M$, $\vs{\alpha} \in V_M$, and $\{\e{T}^x\}_{x\in \Sigma}$ is a family of tensors where each $\e{T}^x\in \bigotimes^{\sharp x} V_M$. 

% Let $M = \langle V_M, \{\e{T}^x\}_{x\in \Sigma}, \odot, \alpha \rangle$ be a HWM. 
Let $G = (V,E,l)$ be a hypergraph and let $\Gamma=[d]^{P_G}$ be the set of mappings from $P_G$ to $[d]$. The \emph{series} $r_M$ computed by the HWM $M$ is defined by $$r_M(G)=\sum_{\gamma\in\Gamma}\e{T}_{\gamma}\prod_{h\in E}\vs{\alpha}^{\top}\bigodot_{i\in\gamma(h)}\v{e}_i$$ where 
$\e{T}_{\gamma}=\prod_{v\in V}\e{T}^v_{\gamma(v^{(1)})\dots\gamma(v^{(\sharp v)})}$ (using the notation $\e{T}^v = \e{T}^{l(v)}$).

% Let $G = (V,E,l)$ be a hypergraph with $V = \{v_1, \cdots, v_M\}$ and $E=\{h_1, \cdots, h_N\}$. 

% Let $\eta : H_E \to \left[ |H_E| \right]$ be the mapping defined by $\eta(v_i, j)  = \sum_{k=1}^{i-1} a_{l(v_k)} + j$, and let $J_i = \{\eta(v,j) : (v,j)\in h_i\}$ for $i\in [N]$. For a tensor $\e{T} \in \bigotimes^k V_A$ where $k \geq |H_E|$, we note 
% $$\left[\e{T}\right]_E = \left[\e{T}\right]_{J_1,\cdots,J_N}$$

% Then, the \emph{rational series} $r_A$ computed by the HWM $A$ is defined by
% $$ r_A(G) =\alpha \left( \left[ \bigotimes_{i=1}^M \e{A}^{l(v_i)} \right]_E \right).$$

\end{definition}

Let % $G = (V,E,l)$ where 
$V = \{v_1, \cdots, v_n\}$. The tensor $\e{T}^{v_1} \otimes
\e{T}^{v_2} \otimes \cdots \otimes \e{T}^{v_n}$ is of order $|P_G|$,
and any element $\gamma \in \Gamma$ can be seen as a multi-index of
$[d]^{|P_G|}$. Thus,
$\e{T}_\gamma$ is the $\left(\gamma(v_1^{(1)}), \cdots,  \gamma(v_1^{(\sharp
  v_1)}), \cdots, \gamma(v_n^{(1)}), \cdots,  \gamma(v_n^{(\sharp
  v_n)})\right)$-coordinate of the tensor $\bigotimes_{i=1}^n  \e{T}^{v_i}$.
% In fact, we have
% $\e{T}_\gamma = \left( \bigotimes_{i=1}^n  \e{T}^{v_i} \right)_{\gamma(v_1^{(1)}), \cdots,  \gamma(v_1^{(\sharp v_1)}), \cdots, \gamma(v_n^{(1)}), \cdots,  \gamma(v_n^{(\sharp v_n)})}$.

\begin{example}
\label{def_hwm_ex1}
Consider the hypergraph $G$ from Example~\ref{ex-1}. We have

$$r_M(G) = \sum_{i_1,\cdots,i_8}\e{T}^a_{i_1i_2i_3}\e{T}^b_{i_4i_5}\e{T}^a_{i_6i_7i_8}\vs{\alpha}^{\top}(\v{e}_{i_1}\odot \v{e}_{i_8}) \vs{\alpha}^{\top}(\v{e}_{i_2}\odot \v{e}_{i_4}\odot \v{e}_{i_7}) \vs{\alpha}^{\top}(\v{e}_{i_3}\odot \v{e}_{i_5}) \vs{\alpha}^{\top}\v{e}_{i_6}.$$
\end{example}

% \begin{remark}
% \label{rmk_HWM_pure_tensors}
% If each tensor $\e{A}^x$ for $x \in \Sigma$ is a pure tensor $\e{A}^x = \v{v}^x_1 \otimes \cdots \v{v}^x_{a(x)}$, then 
% $$ r_A(G) = \alpha \left( \bigotimes_{h\in E} \left( \bigodot_{(v,i) \in h} \v{v}_i^{l(v)} \right)\right) = \prod_{h\in E} \alpha \left( \bigodot_{(v,i) \in h} \v{v}_i^{l(v)} \right) $$
% \end{remark}

% \begin{definition}
% Given a HWM $A = \langle V_A, \{\e{A}^x\}_{x\in \Sigma}, \odot, \alpha \rangle$, we will also map each hypergraph $G=(V,E)$ with $k$ free ports (i.e. there are exactly $k$ singletons in $E$) to a tensor $\e{A}^G \in \bigotimes^k V$. Using the notations from Defintion~\ref{def_hwm}:
% $$ \e{A}^G =  \alpha_{\min J_1}\left( \alpha_{\min J_2} \left( \cdots \alpha_{\min J_k}\left( \left[ \bigotimes_{i=1}^M \e{A}^{l(v_i)} \right]_E \right) \cdots \right) \right) $$ 
% It is easy to check that $r_A(G) = \alpha(\e{A}^G)$.
% \end{definition}

\begin{remark}
If $\odot=\odot_{id}$ and if $\vs{\alpha}=\mathbf{1}$, then $r_M(G)=\sum_{\gamma\in\Gamma_{Id}}\e{T}_{\gamma}$ 
where  $\Gamma_{Id}=\{\gamma\in\Gamma: \forall h\in E, p,q\in h\Rightarrow \gamma(p)=\gamma(q)\}$. For the hypergraph $G$ from Example~\ref{ex-1}, this would lead to the following contractions of the tensor $\e{T}^a \otimes \e{T}^b \otimes \e{T}^a$:
\begin{align*}
r_M(G) &= \sum_{i_1,i_2,i_3,i_6}\e{T}^a_{i_1i_2i_3}\e{T}^b_{i_2i_3}\e{T}^a_{i_6i_2i_1} 
\end{align*}
\end{remark}

\begin{remark}
\label{rmk_circular_strings}
Let $\Sigma$ be a finite alphabet, let $\v{M}_\sigma \in
\K^{d\times d}$ for $\sigma \in \Sigma$ and let $A = \langle \K^d,
\{\v{M}_\sigma\}_{\sigma \in \Sigma}, \odot_{id}, \v{1}\rangle$ be a HWM. For
any non empty word $w = w_1 \cdots w_n \in \Sigma^*$ and its corresponding circular string $G_w$, it follows from Remark~\ref{rmk_contraction} that $r_A( G_w ) = Tr(\v{M}_{w_1} \cdots \v{M}_{w_n})$ (where $Tr(\v{M})$ is the trace of the matrix $\v{M}$). 
\end{remark}

\begin{remark}
\label{rmk_pure_tensors}
Let $A = \langle \R^d, \{\e{A}^x\}_{x\in \Sigma}, \odot, \vs{\alpha} \rangle$ be a HWM. Each tensor $\e{A}^x$ can be decomposed as a sum of rank one tensors $\e{A}^x = \sum_{r=1}^R \v{a}_r^{(x,1)}\otimes \cdots \otimes \v{a}_r^{(x,\sharp x)}$ where $R$ is the maximum rank of the tensors $\e{T}^x$ for $x \in \Sigma$. The computation of the HWM $A$ on $G =(V,E,\ell)$ can then be written as  
$ r(G) = \prod_{h\in E} \vs{\alpha} ^\top \left[ \bigodot_{(v,i)\in h} \left(\sum_{r=1}^R \v{a}_r^{(\ell(v), i)}\right) \right].$
\end{remark}

\begin{remark}
\label{prop_connected_components}
If $G$ is a hypergraph with two connected components $G_1$ and $G_2$, we have $r_M(G) = r_M(G_1)\cdot r_M(G_2)$ for any HWM $M$.
\end{remark}

\begin{definition}
Let $\mathcal{H}$ be a family of hypergraphs on a ranked alphabet $(\Sigma, \sharp)$.
We say that a hypergraph series $r: \mathcal{H} \to \K$ is recognizable if and only if there exists a HWM $M$ such that $r_M(G) = r(G)$ for all  $G \in \mathcal{H}$.
\end{definition}
% \begin{example}
% For the hypergraph $G_0$ and the HWM $A$ from Example~\ref{def_hwm_ex1}, we have
% $$ V\ni \e{A}^{G_0} = \sum_{i_1,\cdots,i_8} a_{i_1,i_2,i_3}b_{i_4,i_5}a_{i_7,i_8} \alpha(\v{e}_{i_1}\odot \v{e}_{i_8}) \alpha (\v{e}_{i_2}\odot \v{e}_{i_4} \odot \v{e}_{i_7})  \alpha (\v{e}_{i_3}\odot \v{e}_{i_5}) \v{e}_{i_6}$$
% Furthermore, if $\odot = \odot_{id}$ and $\alpha$ is s.t. $\alpha(\v{e}_i)=1$ for all $i$, we have the usual contraction
% $$\e{A}^{G_0} = \sum_{i_1,\cdots,i_4} a_{i_1,i_2,i_3}b_{i_2,i_3}a_{i_2,i_1} \v{e}_{i_4}$$

% \end{example}

%%% Local Variables: 
%%% mode: latex
%%% TeX-master: "main"
%%% End: 

%% file: hwm_string_trees.tex
In this section, we show that HWMs satisfy some basic properties which are desirable for a model extending the notion of recognizable series to hypergraphs.

The following propositions show that the proposed model naturally generalizes the notion of linear representation of recognizable series on strings and trees.

\begin{proposition}
Let $r = \langle V, \vs{\iota}, \{\v{M}^\sigma\}_{\sigma\in\Sigma}, \vs{\tau} \rangle$ be a recognizable series on $\Sigma^*$. For any word $w \in \Sigma^*$, let $G_w$ be the associated hypergraph on the ranked alphabet $(\Sigma\cup \{\iota,\tau\}, \sharp)$, whose construction is described in Example~\ref{def_hypergraphs_stringEx}. Consider the HWM $M = \langle V, \{\e{T}^x\}_{x\in \Sigma \cup \{\iota, \tau\}}, \odot_{id}, \v{1}\rangle$ where $\e{T}^\tau = \vs{\tau}$, $\e{T}^\iota = \vs{\iota}$ and $\e{T}^\sigma = \v{M}^\sigma$ for all $\sigma \in \Sigma$.

Then, $r(w) = r_M(G_w)$ for all strings $w\in \Sigma^*$.

\end{proposition}
\begin{proof}
Let $w = w_1\cdots w_n$. We have
\begin{align*}
r(w)&=\vs{\iota}^\top  \v{M}^{w_1} \cdots\v{M}^{w_n} \vs{\tau}=\sum_{i_0, \dots, i_n}\vs{\iota}_{i_0}\v{M}^{w_1}_{i_0,i_1}\dots \v{M}^{w_1}_{i_{n-1},i_n}\vs{\tau}_{i_n}\\
&=\sum_{i_0, \dots, i_n}\e{T}^\iota_{i_0}\e{T}^{w_1}_{i_0,i_1}\dots \e{T}^{w_n}_{i_{n-1},i_n}\e{T}^\tau_{i_n}=r_M(G_w). 
\end{align*}
\end{proof}

In the previous proposition, the vectors  $\vs{\iota}$ and $\vs{\tau}$  of a linear representation were directly encoded in the structure of the graph representation of a string $w$ on $\Sigma$ using the new symbols $\iota$ and $\tau$.  The next proposition shows that it is possible to encode these linear forms in the vector $\vs{\alpha}$ of a HWM with complex coefficients, using a graph representation of strings without new symbols: for any string $w = w_1 \cdots w_n$ over $\Sigma$, we consider the graph $H_w = (V,E,\ell)$ on $(\Sigma, \sharp)$ where $V = [n]$, $\ell(i) = w_i$ and the set of hyperedges is composed of $\{(1,1)\}$, $\{(n,2)\}$ and $\{(i, 2),(i+1,1)\}$ for $i \in [n-1]$ (note that the graph representation of a string is different from the graph representation of its mirror because of the identification of the ports).

\begin{proposition}
\label{prop_iota_eq_tau}

Let $r = \langle \R^d, \vs{\iota}, \{\v{M}_\sigma\}_{\sigma \in \Sigma}, \vs{\tau}\rangle$  be a recognizable string series on $\Sigma^*$. There exists a HWM $M=\langle \C^d, \{\e{T}^\sigma\}_{\sigma\in \Sigma}, \vs{\alpha}, \odot \rangle$ such that $r_M(H_w) = r(w)$ for all $w\in \Sigma^*$.
\end{proposition}

\begin{proof}
We first show that given a recognizable string series $r=\langle \R^d, \vs{\iota}, \{\v{M}^\sigma\}_{\sigma \in \Sigma}, \vs{\tau}\rangle$ there exists a recognizable series $s=\langle \C^d, \vs{\alpha}, \{\v{N}^\sigma\}_{\sigma \in \Sigma}, \vs{\alpha}\rangle$ such that $s(w) = r(w)$ for all $w \in \Sigma^*$. Indeed, let $(\v{e}_1, \cdots, \v{e}_d)$ be a basis of $\R^d$ such that $\v{e}_i^\top \vs{\tau} \not= 0$ and  $\v{e}_i^\top \vs{\iota} \not= 0$ for all $i \in [d]$. Let $\v{D}\in \C^{d\times d}$ be the diagonal matrix defined by $\v{D}_{ii} = (\v{e}_i^\top \vs{\tau})^{1/2} /  (\v{e}_i^\top \vs{\iota})^{1/2}$. We have $\v{D}^\top \vs{\iota} = \v{D}^{-1} \vs{\tau}$ and the series $s:  \langle \C^d, \v{D}^\top\vs{\iota}, \v{D}^{-1} \vs{\tau}, \{\v{D}^{-1}\v{M}^\sigma\v{D}\}_{\sigma \in \Sigma}\rangle$ is such that $s(w) = r(w)$ for all $w \in \Sigma^*$.

We then have that the HWM $M=\langle \C^d, \{\e{T}^x\}_{x\in \Sigma}, \vs{\alpha}, \odot\rangle$, where $\e{T}^\sigma = \v{M}^\sigma$, $\vs{\alpha} =  \v{D}^\top \vs{\iota} = \v{D}^{-1} \vs{\tau}$ and $\odot$ is defined by $\v{e}_i \odot \v{e}_j = \delta_{ij} \frac{1}{\vs{\alpha}_i} \v{e}_i$, is such that $r_M(H_w) = r(w)$ for all $w\in \Sigma^*$.
\end{proof}
% \begin{align*}
% r_A(G_w) &= \alpha\left( \left[ \vs{\iota} \otimes \v{M}_{w_1} \otimes \cdots \otimes \v{M}_{w_n} \otimes \vs{\tau} \right]_{\{1,2\},\{3,4\}, \cdots, \{2n+1,2n+2\}} \right) \\
% &= \alpha\left( \left[ \vs{\iota} \otimes \v{M}_{w_1} \otimes \cdots \otimes \v{M}_{w_{n-1}} \otimes \alpha_2 \left( \left[  \v{M}_{w_n} \otimes \vs{\tau} \right]_{\{2,3\}} \right) \right]_{\{1,2\}, \cdots, \{2n-1,2n\}} \right)\\
% &= \alpha\left( \left[ \vs{\iota} \otimes \v{M}_{w_1} \otimes \cdots \otimes \v{M}_{w_{n-1}} \otimes ( \v{M}_{w_n} \vs{\tau}) \right]_{\{1,2\},\cdots, \{2n-1,2n\}} \right)\\
% &= \cdots = \vs{\iota}^\top  \v{M}_{w_1} \cdots\v{M}_{w_n} \vs{\tau}  = r_s(w)
% \end{align*}

\begin{proposition}
Let $r= \langle V, \mu, \lambda\rangle$ be a recognizable series on trees on the ranked alphabet $\mathcal{F} =(\Sigma, \sharp)$. For any tree $t$ over  $\mathcal{F}$, let $G^t = (V_t,E_t)$ be the associated hypergraph on the ranked alphabet $(\Sigma\cup\{\lambda\},\sharp ')$, whose construction is described in Example~\ref{def_hypergraphs_treeEx}.
Consider the HWM $M: \langle V, \{\e{T}^x\}_{x\in \Sigma \cup \{\lambda\}}, \odot_{id}, \v{1}\rangle$ where $\e{T}^\lambda = \lambda$ and $\e{T}^f$ is defined by $\e{T}^f_{i_0\dots i_k}=\v{e}_{i_0}^{\top}\mu(f)(\v{e}_{i_1},\dots,\v{e}_{i_k})$ for all $k$ and  $f \in \mathcal{F}_k$. 

Then, $r(t) = r_M(G^t)$ for all tree $t$ over $\mathcal{F}$.
\end{proposition}
\begin{proof}
For any $\gamma\in\Gamma_{Id}$, let $\e{U}_{\gamma}=\prod_{v \in V_t \setminus \{0\}}\e{T}^{l(v)}_{\gamma(v,1)\dots\gamma(v,\sharp v)}$. We first prove by induction on $t$ that $\mu(t)=\sum_{\gamma\in\Gamma_{Id}}\e{U}_{\gamma}\v{e}_{\gamma(\epsilon,1)}$.

If $t=a$, then $\sum_{\gamma\in\Gamma_{Id}}\e{U}_{\gamma}\v{e}_{\gamma(\epsilon,1)}=\sum_{i\in[d]}\e{T}^a_i\v{e}_i=\mu(a)$.

If $t=f(t_1, \dots, t_k)$, first remark that $\Gamma_{Id}=\Gamma_{Id}^{(0)}\times \Gamma_{Id}^{(1)}\times \dots \times \Gamma_{Id}^{(k)}$, where $\Gamma_{Id}^{(0)}$ fixes the values of the port $(\epsilon,1)$ and where $\Gamma_{Id}^{(j)}$, for $j\in[k]$, fixes the values of the ports of the subgraph corresponding to the subtree $t_j$. 

\begin{align*}
\mu(f(t_1, \dots, t_k))&=\mu(f)(\mu(t_1), \dots, \mu(t_k))\\
&=\sum_{i_0, i_1, \dots, i_k}\e{T}^f_{i_0, i_1, \dots, i_n}\prod_{j\in[k]}\v{e}_{i_j}^{\top}\mu(t_j) \v{e}_{i_0}\\
&=\sum_{i_0, i_1, \dots, i_k}\e{T}^f_{i_0, i_1, \dots, i_n}\prod_{j\in[k]}\v{e}_{i_j}^{\top}\left(\sum_{\gamma\in\Gamma^{(j)}_{Id}}\e{U}_{\gamma}\v{e}_{\gamma(j,1)}\right) \v{e}_{i_0}\\
&=\sum_{i_0, i_1, \dots, i_k}\e{T}^f_{i_0, i_1, \dots, i_n}\prod_{j\in[k]}\sum_{\gamma\in\Gamma^{(j)}_{Id},\gamma(j,1)=i_j}\e{U}_{\gamma} \v{e}_{i_0}\\
&=\sum_{i_0,\gamma_1\in\Gamma^{(1)}_{Id}, \dots, \gamma_k\in\Gamma^{(k)}_{Id}}\e{T}^f_{i_0, \gamma_1(1,1), \dots, \gamma_k(k,1)}\prod_{j\in[k]}\e{U}_{\gamma_j} \v{e}_{i_0} =\sum_{\gamma\in\Gamma_{Id}}\e{U}_{\gamma} \v{e}_{\gamma(\epsilon,1)}.
\end{align*}

It is then easy to check that $$\lambda(\mu(t))=\sum_{i}\sum_{\gamma\in\Gamma_{Id}}\e{T}^{\lambda}_{i}\e{U}_{\gamma}\v{e}_i^T \v{e}_{\gamma(\epsilon,1)}=\sum_{\gamma\in\Gamma_{Id}}\e{T}^{\lambda}_{\gamma(\epsilon,1)}\e{U}_{\gamma}=\sum_{\gamma\in\Gamma_{Id}}\e{T}_{\gamma}=r_M(t).$$
\end{proof}

%% file: hwm_addition.tex
The following propositions show that the set of HWMs is closed under addition and Hadamard product. 

\begin{proposition}
\label{stability_addition}
Let $A=\langle \K^m, \{\e{A}^x\}_{x\in \Sigma} , \odot_A, \vs{\alpha} \rangle$, and $B=\langle \K^n, \{\e{B}^x\}_{x\in \Sigma} , \odot_B, \vs{\beta} \rangle$ be two HWMs. Let $r_A$ (resp. $r_B$) be the series computed by $A$ (resp. by $B$).% Let $(\v{a}_1, \cdots, \v{a}_m)$ be a basis of $V_A$ and  $(\v{b}_1, \cdots, \v{b}_m)$ be a basis of $V_B$
% Let $V_C$ be a $(m+n)$-dimensional vector space with basis  $(\v{c}_1, \dots, \v{c}_{m+n})$. For any integer $k$, we define the linear embeddings $f_A:\bigotimes^k V_A \to  \bigotimes^{k} V_C$ and $ f_B:\bigotimes^k V_B \to \bigotimes^{k} V_C$ by 
% $$f_A( \v{a}_{i_1}\otimes \cdots \otimes \v{a}_{i_k})= \v{c}_{i_1}\otimes \cdots \otimes\v{c}_{i_k} \mbox{ and } f_B( \v{b}_{i_1}\otimes \cdots \otimes \v{b}_{i_k})= \v{c}_{m+i_1}\otimes \cdots \otimes\v{c}_{m+i_k}$$
 Define the HWM $C = \langle \K^{m+n}, \{\e{C}^x\}_{x\in\Sigma}, \odot, \tau \rangle$ by
\begin{itemize}
\item $ \e{C}^x_{i_1\dots i_{\sharp x}}= \left\{
    \begin{array}{ll}
      \e{A}^x_{i_1\dots i_{\sharp x}}&\textrm{ if }1\leq i_1, \dots, i_{\sharp x}\leq m\\
      \e{B}^x_{j_1\dots j_{\sharp x}}&\textrm{ if }m< i_1, \dots,
      i_{\sharp x}\leq m+n\textrm{ where }j_k=i_k-m\textrm{ for any }k\\
0&\textrm{otherwise,}
    \end{array}
\right.$
\item $\vs{\tau}_i = \vs{\alpha}_i$ if $1\leq i \leq m$ and $\vs{\beta}_{i-m}$ otherwise, and
\item $\v{e}_i \odot \v{e}_j = \begin{cases} \v{e}_i \odot_A \v{e}_j &\mbox{if } 1\leq i,j\leq m \\
    t_m(\v{e}_{i-m} \odot_B \v{e}_{j-m}) &\mbox{if } m< i,j\leq n \\
    0 & \mbox{otherwise } \end{cases}$
\end{itemize}
where $t_m:\K^n\rightarrow \K^{m+n}$ is the linear mapping defined by
$t_m(\v{e}_k)=\v{e}_{k+m}$ for any $1\leq k\leq n$.

% Let $r_A$ be the series computed by $A$, and let $r_{B}$ be the series computed by $B$.
Then the HWM $C$ computes the series $r_{A+B}$ defined by $r_{A+B}(G)=r_{A}(G)+r_{B}(G)$, for any connected hypergraph $G$.
\end{proposition}
\begin{proof}
Let $P_G$ be the set of ports of $G$, let $\Gamma=[m+n]^{P_G}$, $\Gamma_1=\{\gamma\in \Gamma: \gamma(P_G)\subseteq [m]\}$ and  $\Gamma_2=\{\gamma\in \Gamma: \gamma(P_G)\subseteq \{m+1, \dots, m+n\}\}$. 

If $\gamma\not \in \Gamma_1\cup \Gamma_2$, then $\e{C}_{\gamma}\prod_{h\in P_G}\tau^{\top}\bigodot_{i\in\gamma(h)}\v{e}_i=0$. Indeed, let $V_1=\{v\in V:\exists v^{(i)}\in P_G\textrm{ s.t. } \gamma(v^{(i)})\leq m\}$ and $V_2=\{v\in V:\exists v^{(i)}\in P_G\textrm{ s.t. } \gamma(v^{(i)})> m\}$. Note that $V_1$ and $V_2$ are not empty.
\begin{itemize}
\item If there exists $v\in V_1\cap V_2$, then $\e{C}^v_{\gamma(v,1)\dots \gamma(v,\sharp v)}=0$ and therefore $\e{C}_{\gamma}=0$
\item If $V_1 \cap V_2=\emptyset$, there exists a hyperedge $h$ and ports $v_1^{(i)}, v_2^{(j)}\in h$ such that $v_1\in V_1$ and $v_2\in V_2$, since $G$ is connected. Then, $\bigodot_{i\in\gamma(h)}\v{e}_i=0$. 
\end{itemize}
Now, 
\begin{align*}
r_C(G)&=\sum_{\gamma\in \Gamma}\e{C}_{\gamma}\prod_{h\in E}\vs{\tau}^{\top}\bigodot_{i\in\gamma(h)}\v{e}_i\\
&=\sum_{\gamma\in \Gamma_1}\e{C}_{\gamma}\prod_{h\in E}\vs{\tau}^{\top}\bigodot_{i\in\gamma(h)}\v{e}_i+\sum_{\gamma\in \Gamma_2}\e{C}_{\gamma}\prod_{h\in E}\vs{\tau}^{\top}\bigodot_{i\in\gamma(h)}\v{e}_i\\
&=\sum_{\gamma\in \Gamma_A}\e{A}_{\gamma}\prod_{h\in
  E}\vs{\alpha}^{\top}\bigodot_{i\in\gamma(h)}\v{e}_i+\sum_{\gamma\in
  \Gamma_B}\e{B}_{\gamma}\prod_{h\in
  E}\vs{\beta}^{\top}\bigodot_{i\in\gamma(h)}\v{e}_i \\
% &=r_A(G)+\sum_{\gamma'\in
%   \Gamma_B}\e{B}_{\gamma'}\prod_{h\in
%   E}\beta^{\top}\bigodot_{i\in\gamma'(h)}\v{e}_i
\end{align*}
where $\Gamma_A=[m]^{P_G}$ and $\Gamma_B=[n]^{P_G}$.
Eventually, $r_C(G)=r_A(G)+r_B(G).$\end{proof}

%% file: hwm_product.tex
\begin{proposition}
Let $A=\langle \K^m,\{\e{A}^x\}_{x\in \Sigma} , \odot^A, \vs{\vs{\alpha}} \rangle$ and $B=\langle \K^n,\{\e{B}^x\}_{x\in \Sigma} , \odot^B, \vs{\beta} \rangle$ be two HWMs. 

Identifying $\K^m \otimes \K^n$ with $\K^{mn}$ via  the mapping $\v{e}_i  \otimes \v{e}_j \mapsto \v{e}_{n(i-1) + j}$, we define                                                                                                                                                                                                                                                                                                                                                                                                                         the HWM $D = \langle \K^m \otimes \K^n, \{\e{D}^x\}_{x\in\Sigma}, \odot, \vs{\delta} \rangle$  by 
\begin{itemize}
\item $\e{D}^x= \e{A}^x \otimes \e{B}^x$ for all $x\in\Sigma$
\item $(\v{a}_1\otimes\v{b}_1)\odot (\v{a}_2\otimes\v{b}_2) = (\v{a}_1\odot^A\v{a}_2)\otimes  (\v{b}_1\odot^B\v{b}_2)$ for all $\v{a}_1,\v{a}_2\in \K^m$  and $\v{b}_1,\v{b}_2 \in \K^n$
\item $\vs{\delta} = \vs{\alpha} \otimes \vs{\beta}$ (i.e. $\vs{\delta}^\top (\v{a}\otimes \v{b}) = (\vs{\alpha}^\top\v{a})(\vs{\beta}^\top\v{b})$  for any $\v{a}\in \K^m$ and $\v{b}\in \K^m$)
\end{itemize}

Let $r_A$ (resp. $r_B$) be the series computed by $A$ (resp. by $B$). Then the HWM $C$ computes the series $r_{C}(G)=r_{A}(G)r_{B}(G)$, for any hypergraph $G$.
\end{proposition}

\begin{proof}
Let $G=(V,E,l)$ be a hypergraph, and let $\Gamma_k = [k]^{P_G}$ for any integer $k$. We will identify $[m]\times [n]$ with $[mn]$ via the mapping $(i,j) \mapsto n(i-1) + j$, and by extension $\Gamma_{mn}$ with $ \Gamma_{m\times n} = ([m]\times [n])^{P_G}$. For any $\gamma \in \Gamma_{m\times n}$ we will note $(\gamma_1, \gamma_2)$ the only element of $\Gamma_m \times \Gamma_n$ satisfying $\gamma(\cdot) = (\gamma_1(\cdot), \gamma_2(\cdot))$.

First note that for any hyperedge $h\in E$ and $\gamma \in \Gamma_{m\times n}$, we have 
\begin{align*}
\vs{\delta}^\top \bigodot_{(i,j) \in \gamma(h)} (\v{e}_i \otimes \v{e}_j) &= \vs{\delta}^\top \left[ \left( \bigodot^{\quad {}_A}_{i \in \gamma_1(h)} \v{e}_i\right) \otimes \left(\bigodot^{\quad {}_B}_{j \in \gamma_2(h)} \v{e}_j\right)  \right] \\
&= \left( \vs{\alpha}^\top \bigodot^{\quad {}_A}_{i \in \gamma_1(h)} \v{e}_i\right) \left( \vs{\beta}^\top \bigodot^{\quad {}_B}_{j \in \gamma_2(h)} \v{e}_j\right)
\end{align*} 
Then, note that for any $\gamma \in \Gamma_{m \times n}$, we have 
$$\e{D}_\gamma = \prod_{v\in V} \e{D}^{l(v)}_{\gamma(v,1) \cdots \gamma(v,\sharp v)} = \prod_{v\in V} \e{A}^{l(v)}_{\gamma_1(v,1), \cdots, \gamma_1(v,\sharp v)} \e{B}^{l(v)}_{\gamma_2(v,1), \cdots, \gamma_2(v,\sharp v)} = \e{A}_{\gamma_1} \e{B}_{\gamma_2}$$

Finally, we have
\begin{align*}
r_D(G) &= \sum_{\gamma \in \Gamma} \e{D}_\gamma \prod_{h\in E} \vs{\delta}^\top \bigodot_{(i,j) \in \gamma(h)} (\v{e}_i \otimes \v{e}_j) \\
&= \sum_{\gamma_1 \in \Gamma_m} \sum_{\gamma_2 \in \Gamma_n} \e{A}_{\gamma_1} \e{B}_{\gamma_2} \prod_{h\in E}  \left( \vs{\alpha}^\top \bigodot^{\quad {}_A}_{i \in \gamma_1(h)} \v{e}_i\right) \left( \vs{\beta}^\top \bigodot^{\quad {}_B}_{j \in \gamma_2(h)} \v{e}_j\right)
= r_A(G) r_B(G)
\end{align*}
\end{proof}

Finally, the next proposition shows that any recognizable real valued series on closed graphs can be computed by a HWM with coefficients in $\C$ using the identity product $\odot_{id}$ and the vector $\mathbf{1}$.

\begin{proposition}
Let $A = \langle \R^d,\{\e{A}^x\}_{x\in \Sigma} , \odot_A, \vs{\vs{\alpha}} \rangle$ be a HWM. There exists a HWM $B = \langle \C^d,\{\e{B}^x\}_{x\in \Sigma} , \odot_{id}, \mathbf{1} \rangle$ such that $r_B(G) = r_A(G)$ for any closed graph $G$.
\end{proposition}

\begin{proof}
We consider the decomposition $\e{A}^x = \sum_{r=1}^R \v{a}^{(x,1)}_r \otimes \cdots \otimes \v{a}^{(x,\sharp x)}_r$ for each $x\in \Sigma$ (see Remark~\ref{rmk_pure_tensors}). Let $\v{M}\in \R^{d\times d}$ be the matrix defined by $\v{M}_{ij} = \vs{\alpha}^\top (\v{e}_i \odot_A \v{e}_j)$ and check that $\v{u}^\top \v{M} \v{v} = \vs{\alpha}^\top (\v{u} \odot_A \v{v})$ and $\mathbf{1}^\top (\v{u} \odot_{id} \v{v}) = \v{u}^\top \v{v}$ for any $\v{u},\v{v} \in \R^d$. Let $\v{Q} \in \C^{d\times d}$ be such that $\v{M} = \v{Q}^\top \v{Q}$ (such a decomposition exists since $\v{M}$ is symmetric) and let $\e{B}^x = \sum_{r=1}^R (\v{Q}\v{a}^{(x,1)}_r) \otimes \cdots \otimes (\v{Q}\v{a}^{(x,\sharp x)}_r)$. For any closed graph $G = (V,E,\ell)$, it follows from Remark~\ref{rmk_pure_tensors} that

\begin{align*}
r_B(G) &= \prod_{h\in E,\atop h=\{(v,i),(w,j)\}} \mathbf{1}^\top \left[ \left(\sum_{r=1}^R \v{Q}\v{a}^{(\ell(v),i)}_r\right) \odot_{id}  \left(\sum_{r=1}^R \v{Q}\v{a}^{(\ell(w),j)}_r\right) \right] \\
&= \prod_{h\in E,\atop h=\{(v,i),(w,j)\}} \left( \sum_{r=1}^R \v{a}^{(\ell(v),i)}_r \right)^\top \v{Q}^\top \v{Q} \left(\sum_{r=1}^R \v{a}^{(\ell(w),j)}_r\right) \\
&= \prod_{h\in E,\atop h=\{(v,i),(w,j)\}} \vs{\alpha}^\top \left[ \left(\sum_{r=1}^R \v{a}^{(\ell(v),i)}_r\right) \odot_{A}  \left(\sum_{r=1}^R \v{a}^{(\ell(w),j)}_r\right) \right] = r_A(G).
\end{align*}
\end{proof}

%%% Local Variables: 
%%% mode: latex
%%% TeX-master: "main"
%%% End: 

%% file: hwm_tilings_motivation.tex
In this section, we show that finite support series (or \emph{polynomials}: series for which the set of hypergraphs with non-zero value is finite) are not recognizable in general, but we exhibit a wide class of families of hypergraphs for which they are.

First, we show on a simple example why polynomials are not recognizable for all families of hypergraphs. Consider the family of circular strings over a one letter alphabet $\Sigma = \{a\}$ introduced in Example~\ref{ex_circular_strings} and Remark~\ref{rmk_circular_strings}. The following lemma implies that the series $r$, defined by $r(G_a) = 1$ and $r(G_{a^k})=0$ for all integer $k > 1$, is not recognizable. 
Indeed, $r$ would be such that $r(G_{a^k}) = Tr(\v{M}_a^{k}) = 0$ for all $k \geq 2$, but it then follows from Lemma~\ref{lemma_circular_strings} that $r(G_{a}) = Tr(\v{M}_a)=0$. 
\begin{lemma}
\label{lemma_circular_strings}
Let $\v{M}\in \R^{n\times n}$. If $Tr(\v{M}^k)=0$ for all $k \geq 2$, then $Tr(\v{M})=0$.
\end{lemma}
\begin{proof}

Let $\lambda_1, \dots, \lambda_p\in \C$ be the distinct non zero
eigenvalues of $\v{M}$, with multiplicities $n_1, \dots, n_p$. The matrix
$\v{N}=(\lambda_i^j)_{i\in[p],2\leq j\leq p+1}$ is full rank: its determinant
is equal to $ \prod_i\lambda_i^2\prod_{i<j}(\lambda_j-\lambda_i)$. We have $( n_1, \dots, n_p)\v{N}=\left( Tr(\v{M}^2),\dots,
Tr(\v{M}^{p+1})\right)=(0, \cdots, 0)$. Therefore, the unique eigenvalue of $\v{M}$ is equal to 0 and
hence, $Tr(\v{M})=0$.
\end{proof}

This example illustrates the fact that the computation of a HWM on a hypergraph $G$ is done independently on each hyperedge of $G$. This implies that if two hypergraphs are not distinguishable by just looking at the ports involved in their hyperedges, the computations of a HWM on these two hypergraphs are strongly dependent. This is clear if we consider a hypergraph $G_1$ made of two copies of a hypergraph $G_2$ (i.e. $G_1$ has two connected components, which are both isomorphic to $G_2$): we have $r(G_1) = r(G_2)^2$ for any HWM $r$ (see Remark~\ref{prop_connected_components}). 

The following section formally introduces the notion of \emph{tiling} of a hypergraph $G$ and show how this relation between hypergraphs relates to the question of the recognizability of polynomials.

%%% Local Variables: 
%%% mode: latex
%%% TeX-master: "main"
%%% End: 

%% file: hwm_finiteSupport.tex
A \emph{tiling} of a hypergraph $\widehat{G}$ is a hypergraph $G$, built on the same alphabet and made of copies of $\widehat{G}$. More precisely, 

\begin{definition}
Let $\widehat{G}=(\widehat{V},\widehat{E},\widehat{l})$ be a hypergraph over a ranked alphabet $(\Sigma, \sharp)$. A hypergraph $G=(V,E,l)$ on the same alphabet $(\Sigma, \sharp)$ is a \emph{tiling} of $\widehat{G}$ if and only if there exists a mapping $f:V \to \widehat{V}$ such that
\begin{itemize}
\item[(i)] $l(v) = \widehat{l}(f(v))$ for  any $v\in V$
\item[(ii)] the mapping  $g: P_G \to P_{\widehat{G}}$ defined by $g(v,i) = (f(v),i)$ is such that for all $h\in E$: $g(h) \in \widehat{E}$ and the restriction $g_{|h}$ of $g$ to $h$ is bijective. 
\end{itemize}
\end{definition}

The following proposition shows that for a connected hypergraph, this formal definition of tiling is equivalent to the intuition of a hypergraph made of copies of the original one.

Let $G=(V,E,l)$ be a tiling of the connected hypergraph $\widehat{G}=(\widehat{V},\widehat{E},\widehat{l})$, let $\sim_V$ be the equivalence relation defined on $V$ by $v\sim_V v'$ iff $f(v)=f(v')$, and let $\sim_E$ be the equivalence relation defined on $E$ by $h\sim_E h'$ iff $g(h)=g(h')$ where $f$ and $g$ are the mappings defined above. Clearly, $v\sim_V v'$ entails that $l(v)=l(v')$ and it can easily be shown that  $h\sim_E h'$ iff $\exists v^{(i)}\in h, v'^{(i)}\in h'$ such that $v\sim_Vv'$. We can thus define the quotient hypergraph $\overline{G}=(V/\sim_V, E/\sim_E, l)$. 

\begin{proposition}
If $G=(V,E,l)$ is a tiling of a connected hypergraph $\widehat{G}=(\widehat{V},\widehat{E},\widehat{l})$, then $\overline{G}=(V/\sim_V, E/\sim_E, l)$ is isomorphic to $\widehat{G}$ and moreover, for any $\hat{v}\in \widehat{V}$, the cardinal of $f^{-1}(\{\hat{v}\})$ is a constant. 
\end{proposition}
\begin{proof}
We will prove the last part of the proposition, which entails the surjectivity of $f$. This will be enough since  if $f$ is surjective, then $\overline{G}$ is isomorphic to $\widehat{G}$.

Let $m$ be the maximal cardinality of the sets $f^{-1}(\{\widehat{v}\})$ and suppose that they have different cardinalities. Let $V_1=\{\widehat{v}\in \widehat{V}: Card(f^{-1}(\{\widehat{v}\}))=m\}$ and $V_2=\widehat{V}\setminus V_1$. Since $\widehat{G}$ is connected, there exists a hyperedge $\widehat{h}$ and $\widehat{v}_1^{(i)}, \widehat{v}_2^{(j)}\in\widehat{h}$ such that $\widehat{v}_1\in V_1$ and $\widehat{v}_2\in V_2$. Let $f^{-1}(\{\widehat{v}_1\})=\{v_1, \dots, v_m\}$ and let $h_1, \dots, h_m\in E$ be the hyperedges containing $v_1 ^{(i)}, \dots, v_m ^{(i)}$, respectively. Since each $g_{|h_i}$ is injective and since the vertices $v_1, \dots, v_m$ are distinct, the hyperedges  $h_1, \dots, h_m$ are also distinct and therefore disjoint. Let $w_1^{(j)}=g_{|h_1}^{-1}(\widehat{v}_2 ^{(j)}), \dots, w_m^{(j)} =g_{|h_m}^{-1}(\widehat{v}_2 ^{(j)})$. These ports are distinct and therefore, the vertices $w_1, \dots, w_m$ are also distinct. Since, $f(w_1)=\dots=f(w_m)=\widehat{v}_2$, we obtain a contradiction. 
\end{proof}

\begin{figure}
\begin{center}
\includegraphics[scale=1]{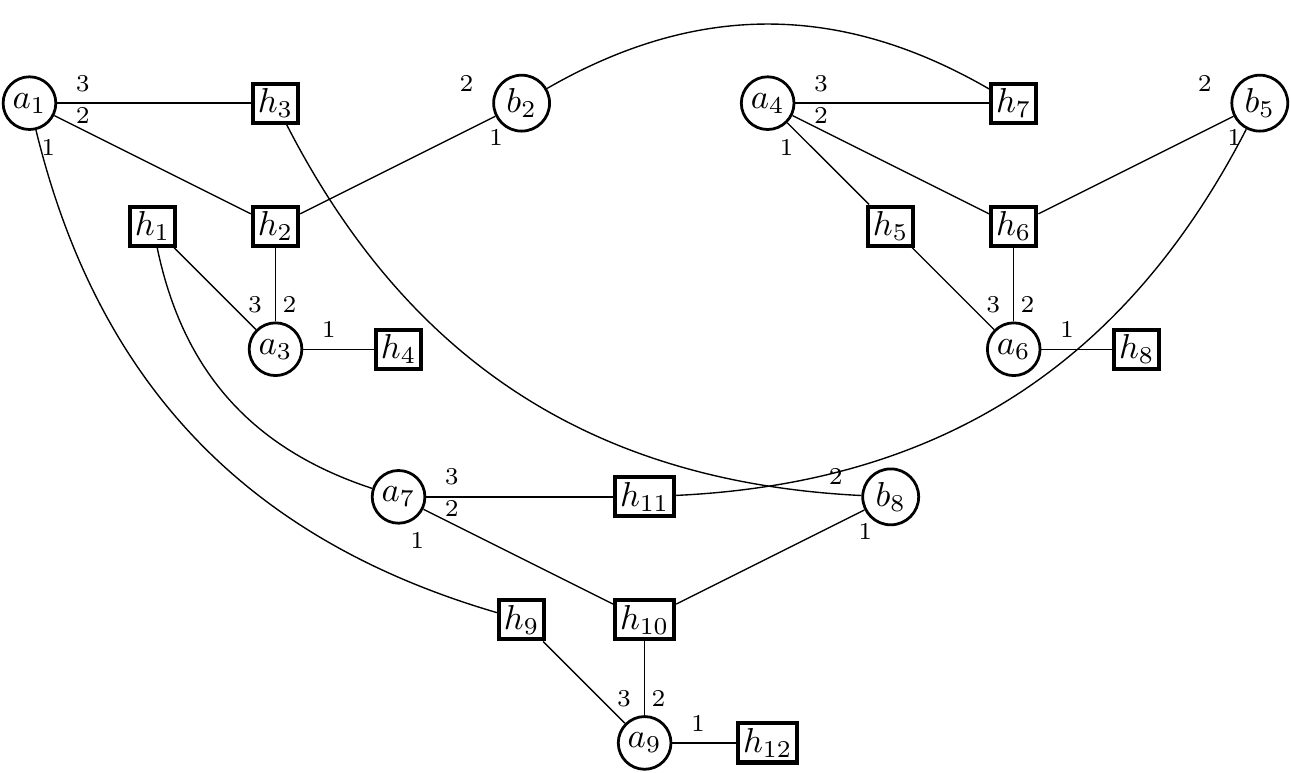}

\caption{A tiling made of three copies of the hypergraph from Example~\ref{ex-1}}
\vspace{-0.25cm}
\end{center}

\end{figure}

\subsection{Finite Support Series and Tilings}
We end this section with the main result of this paper. We show that we can construct a HWM which assigns a nonzero value to a specific hypergraph over some ranked alphabet and all of its tilings, and zero to any other hypergraph on the same alphabet. This result leads to a sufficient condition on families of hypergraphs for the recognizability of finite support series.

\begin{theorem}
Given a hypergraph $\widehat{G} = (\widehat{V}, \widehat{E}, \widehat{l})$ over $(\Sigma,\sharp)$, there exists a recognizable series $r_{\widehat{G}}$ such that $r_{\widehat{G}} (G) \not = 0$ if and only if $G$ is a tiling of $\widehat{G}$.
\end{theorem}
\begin{proof}
Let $P_{\widehat{G}}$ be the set of ports of $\widehat{G}$. For any symbol $x\in \Sigma$, we note $\widehat{V}(x)$ the set of vertices in $\widehat{V}$ labelled by $x$.

Let $\mathcal{S} =  2^{P_{\widehat{G}}}$ bet the set of subsets of $P_{\widehat{G}}$ and let $d = |\mathcal{S}|$. Instead of indexing the canonical basis of $\K^d$ with integers in $[d]$, we will index it with elements of $\mathcal{S}$. For example, for each  port $(\hat{v},i) \in P_{\widehat{G}}$, the singleton $\{(\hat{v},i)\}$ is in $\mathcal{S}$, thus $\v{e}_{\{(\hat{v},i)\}}$ is a basis vector (which we will note $\v{e}_{(\hat{v},i)}$ for convenience).

Define the HWM $M=\langle \K^d,\{\e{T}^x\}_{x\in \Sigma} , \odot, \vs{\alpha} \rangle$ by
\begin{align*}
 \e{T}^x &= {\begin{cases} 
 {\v{e}_\emptyset} ^{\otimes \sharp x} &\mbox{if } \widehat{V}(x) = \emptyset  \\
  \sum_{\hat{v}\in \widehat{V}(x)} \v{e}_{(\hat{v},1)} \otimes \cdots \otimes \v{e}_{(\hat{v},\sharp \hat{v})} &\mbox{otherwise}
\end{cases} }\\
\v{e}_S \odot \v{e}_T &= {\begin{cases}
\v{e}_{S\cup T} &\mbox{if } S \not = \emptyset, T \not = \emptyset \mbox{ and } S \cap T = \emptyset \\
\v{e}_\emptyset &\mbox{otherwise}
\end{cases}} \\
\vs{\alpha}_S &= {\begin{cases}
1 &\mbox{if } S \in \widehat{E} \quad (\mbox{note that } \emptyset \not \in \widehat{E}) \\
0 & \mbox{otherwise}
\end{cases}}
\end{align*}
for any $x\in \Sigma$ and $S,T \in \mathcal{S}$. Let $r$ be the series computed by $M$, we claim that $r$ satisfies the property of the theorem.

For any hypergraph $G = (V,E,l)$ with $V = \{v_1, \cdots, v_N\}$, we have
$$
r(G) = \sum_{\gamma\in\Gamma}\e{T}_{\gamma}\prod_{h\in E}\vs{\alpha}^{\top}\bigodot_{S \in\gamma(h)}e_S
$$
where $\Gamma = \mathcal{S}^{P_G}$ and $\e{T}_{\gamma} = \prod_{i=1}^N \e{T}^{v_i}_{\gamma(v_i, 1), \cdots, \gamma(v_i, \sharp v_i)}$. Let $ \gamma \in \Gamma$. If there exists a port $p \in P_G$ such that $\gamma(p) = \emptyset$, then $\prod_{h\in E}\vs{\alpha}^{\top}\bigodot_{S \in\gamma(h)}e_S = 0$; otherwise, it follows from the definition of the tensors $\e{T}^x$ that  $\e{T}_{\gamma} $ is different from $0$ if and only if for all $v\in V$, there exists $\widehat{v} \in \widehat{V}(l(v))$ s.t. $\gamma(v,i) = (\widehat{v},i)$ for all $i \in [\sharp v]$, hence
$$
r(G) = \sum_{\widehat{v}_1\in \widehat{V}(l(v_1))} \cdots  \sum_{\widehat{v}_N\in \widehat{V}(l(v_N))}   \prod_{i=1}^N \e{T}^{v_i}_{(\widehat{v}_i,1) \cdots (\widehat{v}_i,\sharp \widehat{v}_i))} \prod_{h\in E}\vs{\alpha}^{\top}\bigodot_{(v_j,i_j) \in h}\v{e}_{(\widehat{v}_j, i_j)}
$$

We have $r(G) \not = 0$ if and only if there exist $N$ vertices $\widehat{v}_i \in \widehat{V}(l(v_i))$ for $i \in [N]$ such that (i) $\vs{\alpha}^\top  \bigodot_{(v_j, i_j) \in h} \v{e}_{(\widehat{v}_j, i_j)} \not = 0$ for all $h \in E$. Let $f: V \to \widehat{V}$ and $g:P_G \to P_{\widehat{G}}$ be the mappings defined by $f(v_i)=\widehat{v}_i$ and $g(v_i,j) = (\widehat{v}_i,j)$ for all $i \in [N]$. It follows from the definitions of $\vs{\alpha}$ and $\odot$ that (i) is true if and only if  $g(h) \in \widehat{E}$ for all $h \in E$, and there are no distinct $(v_j,i_j),(v_k,i_k)$ in a hyperedge $h\in E$ such that $(\widehat{v}_j,i_j) = (\widehat{v}_k,i_k)$, i.e. the restriction $g_{|h}$ is injective for any $h \in E$. 
\end{proof}

A family $\mathcal{H}$ of hypergraphs is \emph{tiling-free} if and only if for any $G\in \mathcal{H}$, there are no (non-trivial) tiling of $G$ in $\mathcal{H}$. 

\begin{corollary}
For any tiling-free family of hypergraphs $\mathcal{H}$, finite support series on $\mathcal{H}$ are recognizable.
\end{corollary}
\begin{proof}
Let $r_{\widehat{G}}$ be the series from the previous proof, and let $z = r_{\widehat{G}}(\widehat{G})$. For any scalar $y$, if we change the definition of $\vs{\alpha}$ to $\vs{\alpha}(\v{e}_S) = (y/z)^{1/|\widehat{E}|}$ if $S \in \widehat{E}$ and 0 otherwise, we have  $r_{\widehat{G}}(\widehat{G}) = y$. The corollary then directly follows from the previous theorem and Proposition~\ref{stability_addition}.
\end{proof}

\begin{remark}
A simple example of tiling-free family is the family of rooted hypergraphs: hypergraphs on a ranked alphabet $(\Sigma \cup \{\lambda\}, \sharp)$, where the special \emph{root} symbol $\lambda$ appears exactly once. 
\end{remark}

%% file: hwm_examples.tex
In this section, we present some examples of recognizable hypergraph series in order to give some insight on the expressiveness of HWMs and on how their computation  relates to the usual notion of recognizable series on strings and trees.

\subsection{Rooted Circular Strings}
First note that the family of circular strings is not tiling-free: the circular string $abab$ is a tiling of $ab$. Instead of the construction described in Example~\ref{def_hypergraphs_stringEx}, we can map each string $w$ on a finite alphabet $\Sigma$ to a \emph{rooted circular string}. Let $w = w_1 \cdots w_n \in \Sigma^*$, we will consider the circular string $G_w$ on the ranked alphabet $(\Sigma \cup \{\lambda\}, \sharp)$ where $\sharp x = 2$ for any $x\in \Sigma \cup \{\lambda\}$, with vertices $V = \{0,\cdots,n\}$, labels $l(0) = \lambda$ and $l(i) = w_i$ for $i \in [n]$, and edges $\{(n,2), (0,1)\}$ and $\{(i, 2), (i+1,1)\}$ for $ i \in \{0,\cdots,n-1\}$ (see Figure~\ref{fig_circular_string}).

Let $r = \langle \R^d, \vs{\iota}, \{\v{M}_\sigma\}_{\sigma \in \Sigma},\vs{\tau} \rangle $ be a rational series on $\Sigma^*$. We define the HWM $A=\langle \R^d, \{\e{A}^x\}_{x\in \Sigma} , \odot_{id}, \v{1} \rangle$ where $\e{A}^\sigma = \v{M}_\sigma$ for all $\sigma \in \Sigma$ and $\e{A}^\lambda = \vs{\iota} \vs{\tau}^\top$. It is easy to check that $r_A(G_w) =r(w)$ for all $w \in \Sigma^*$.

Now consider $m$ rational series on  $\Sigma^*$ with $d$-dimensional linear representations $\langle \vs{\iota}_i, \{\v{M}_\sigma\}_{\sigma \in \Sigma}, \vs{\tau}_i \rangle $ for $i\in [m]$. The string series $ r = r_1 + \cdots + r_m$  is rational and its dimension can be as high as $dm$. However, the HWM $A=\langle \R^d, \{\e{A}^x\}_{x\in \Sigma} , \odot_{id}, \v{1} \rangle$ where $\e{A}^\sigma = \v{M}_\sigma$ for all $\sigma \in \Sigma$ and $\e{A}^\lambda = \sum_{i=1}^m \vs{\iota}_i \vs{\tau}_i^\top$ is such that $r_A(G_w) = r(w)$ for all $w\in \Sigma^*$, and is of dimension $d$.

\subsection{Recognizing $a^n b^n$}
Given an even length string on the alphabet $\{a,b\}$, we can enrich the construction described in Example~\ref{def_hypergraphs_stringEx} by associating a vertex of arity 3 to each letter, and adding extra edges connecting letters in the first half of the string to letters in the second half. Formally, given a word $w_1\cdots w_{2n}$ on $\Sigma$, we consider the  \emph{3-ary graph representation} of $w$ given by the graph $G = (V,E,l)$ on the ranked alphabet $(\Sigma \cup \{\iota, \tau\}, \sharp)$ where $\sharp \iota = \sharp \tau = 1$, $\sharp \sigma = 3$ for all $\sigma \in \Sigma$, $V = \{ 0 \cdots 2n+1\}$, $l(0) =  \iota$, $l(2n+1) = \tau$, $l(i) = w_i$ for $1\leq i \leq 2n$, and the set of edges is composed of $\{(0,1), (1,1)\}$, $\{(i,2), (i+1,1)\}$ for $1 \leq i \leq 2n$ and $\{(i,3), (n+i,3)\}$ for $1\leq i \leq n$.  The 3-ary graph representation of the string $abaa$ is shown in Figure~\ref{fig_circular_string}. 

Using this construction, it is easy to show that there exists a HWM (on the family of 3-ary graph representations of even length strings) whose support is the set of 3-ary graph representations of the language $\{a^nb^n: n\geq 1\}$.

\subsection{Crosswords}
\label{sec_crosswords}

Let $(\Sigma,\sharp)$ be a ranked alphabet where all symbols have
arity 4. An $(M,N)$-crossword $w$ on $\Sigma$ is  an array of symbols
$[w_{ij}]_{ij}\in \Sigma^{M\times N}$. The graph $G_w = (V,E,l)$
associated to the crossword $w$ is the graph with vertices $V =
[M]\times [N]$,  $l(m,n) = w_{mn}$, and edges $E=E_H\cup E_V$, where
the ports are labeled by $W, E, N, S$, in this order, and where % \footnote{recall that $v^i$ denotes the $i$-th port of vertex $v$  for any $v\in V$} 
$$E_H=\bigcup_{m\in [M],n\in [N-1]} \left\{ \{(m,n)^E,(m,n+1)^W\}
\right\}\bigcup_{m\in [M]} \left\{ \{(m,1)^W\},\{(m,N)^E\} \right\}$$
and
$$E_V=\bigcup_{n\in [N],m\in [M-1]} \left\{ \{(m,n)^S,(m+1,n)^N\} \right\}\bigcup_{n\in [N]} \left\{ \{(1,n)^N\},\{(M,n)^S\} \right\}.$$
% \begin{align*}
% &\bigcup_{m\in [M],n\in [N-1]} \left\{ \{(m,n)^E,(m,n+1)^W\} \right\},\ \ \bigcup_{n\in [N],m\in [M-1]} \left\{ \{(m,n)^S,(m+1,n)^N\} \right\}, \\
% &\bigcup_{m\in [M]} \left\{ \{(m,1)^W\},\{(m,N)^E\} \right\},\ \mbox{ and } \bigcup_{n\in [N]} \left\{ \{(1,n)^N\},\{(M,n)^S\} \right\}.
% \end{align*}
\vspace{0.2cm}
An example of graph associated to a 2D word is shown in Figure~\ref{fig_circular_string}.
 Now, let $\#_1$ be a new arity function, such that all symbols of $\Sigma$
have arity 2. The graph $G_w$ can be decomposed in two graphs  over $(\Sigma,\#_1)$:
$G_w^{H}=(V,E_H,l)$ where the ports are labeled $W,E$ in this order
and $G_w^{V}=(V,E_V,l)$ where the ports are labeled $N,S$ in this
order. 

\begin{theorem}\label{th:cw} Given $\Ac = \langle
\R^{d_1},\{\e{A}^\sigma\}_{\sigma\in\Sigma}, \odot_1,
\vs{\beta}_1\rangle$ and 
$\Bc = \langle
\R^{d_2},\{\e{B}^\sigma\}_{\sigma\in\Sigma}, \odot_2,
\vs{\beta}_2\rangle$, two HWMs over $(\Sigma,\#_1)$, there exists a HWM $\Cc =
\langle \R^{d_1+d_2},\{\e{C}^\sigma\}_{\sigma\in\Sigma}, \odot,
\vs{\beta}\rangle$ over $(\Sigma, \sharp)$ such that $$r_{\Cc}(G_w) = r_{\Ac}(G_{w}^H)\times r_{\Bc}(G_{w}^V)$$
for any $(M,N)$-crossword $w$.

\end{theorem}
\begin{proof}
Let $(\v{e}_1, \dots, \v{e}_{d_1}, \v{f}_1, \dots, \v{f}_{d_2})$ be the canonical basis
of $\R^{d_1+d_2}$. The HWM $C$ is defined by
\begin{itemize}
\item $\e{C}^\sigma=\e{A}^\sigma\otimes \e{B}^\sigma=\left(\sum_{i_1, \dots,
    i_{\sharp \sigma}\in [d_1]}\e{A}_{i_1, \dots, i_{\sharp \sigma}}\v{e}_{i_1}\otimes \dots
  \v{e}_{i_{\sharp \sigma}}\right)\otimes \left(\sum_{i_1, \dots,
    i_{\sharp \sigma}\in [d_2]}\e{B}_{i_1, \dots, i_{\sharp \sigma}}\v{f}_{i_1}\otimes \dots
  \v{f}_{i_{\sharp \sigma}}\right)$ for any $\sigma \in \Sigma$
\item $\v{e}_i\odot \v{e}_j=\v{e}_i\odot_1 \v{e}_j, \v{f}_i\odot \v{f}_j=\v{f}_i\odot_2 \v{f}_j$ and
  $\v{e}_i\odot \v{f}_j=\v{f}_j\odot \v{e}_i=0$ for any indices $i,j$
\item $\vs{\beta}= \left( \begin{matrix}
\vs{\beta}_1\\
\vs{\beta}_2
\end{matrix}\right)$. 
\end{itemize}

By definition, we
have $$r_\Cc(G_w)=\sum_{\gamma\in\Gamma}\e{C}_{\gamma}\prod_{h\in
  E}\vs{\beta}^{\top}\bigodot_{i\in\gamma(h)}\v{g}_i$$ where $\v{g}_i = \v{e}_i$ if $i \in [d_1]$ and $\v{f}_i$ otherwise. Let $\Gamma_H = [d_1]^{P_{G_w^H}}$ and $\Gamma_V= [d_2]^{P_{G_w^V}}$. It is easy to check that any $\gamma\in \Gamma$ for which $\e{C}_\gamma \not= 0$ can be associated with a tuple $(\gamma_H,\gamma_V) \in \Gamma_H  \times \Gamma_V$ satisfying
$\e{C}_{\gamma}=\e{A}_{\gamma_H}\e{B}_{\gamma_V}$, we have
\begin{align*}
  r_\Cc(G_w)&=\sum_{\gamma_H\in\Gamma_H}\sum_{\gamma_V\in\Gamma_V}\e{A}_{\gamma_H}\e{B}_{\gamma_V}\prod_{h\in
    E_H}\vs{\beta}_1^{\top} \left( \underset{i\in\gamma_H(h)}{\odot_1}\v{e}_i \right) \times
  \prod_{h\in
    E_V}\vs{\beta}_2^{\top} \left( \underset{i\in\gamma_V(h)}{\odot_2}\v{f}_i\right)   \\
&=r_{\Ac}(G_{w}^H)\times r_{\Bc}(G_{w}^V).
\end{align*}

\end{proof}

Given a $(M,N)$-crossword $w$, we note $w_{m:}$ for the $m$-th row of $w$ and $w_{:n}$ for its $n$-th column.

\begin{corollary}
Let $\Ac = \langle\R^{d_1},  \{\v{A}^\sigma\}_{\sigma\in \Sigma}, \vs{\alpha}_0, \vs{\alpha}_{\infty}\rangle$ and $\Bc = \langle \R^{d_2}, \{\v{B}^\sigma\}_{\sigma\in \Sigma}, \vs{\beta}_0, \vs{\beta}_{\infty}\rangle$ be two rational string series on $\Sigma^*$.

There exists a HWM $\Cc = \langle \C^{d_1+d_2},\{\e{C}^\sigma\}_{\sigma\in\Sigma}, \odot, \vs{\gamma}\rangle$ such that $$r_{\Cc}(G_w) = \prod_{m\in [M]} r_{\Ac}(w_{m:}) \prod_{n\in [N]} r_{\Bc} (w_{:n})$$
for any $(M,N)$-crossword $w$.
\end{corollary}

\begin{proof}
The result directly follows from Proposition~\ref{prop_iota_eq_tau} and Theorem~\ref{th:cw}, and
  remarking that the HWM $\Mc_H$ (resp. $\Mc_V$) that computes
$r_{\Ac}$ (resp. $r_{\Bc}$) satisfies $\Mc_H(G_w^H)=\prod_{m\in
  [M]} r_{\Ac}(w_{m:})$ (resp. $\Mc_V(G_w^V)=\prod_{n\in [N]} r_{\Bc}
(w_{:n}))$ since the graphs $G_w^H$ (resp. $G_w^V$) has $M$
(resp. $N$) connected components.
\end{proof}
\begin{figure}

\vspace{-1cm}
\begin{center}
\includegraphics[scale=1, trim= 0cm -0.75cm 0cm 0cm]{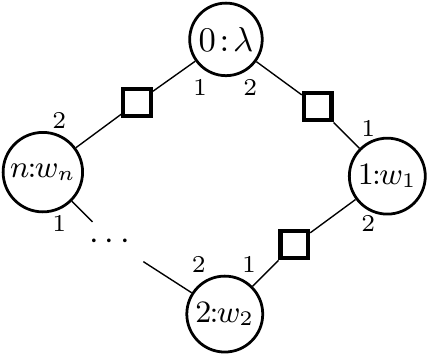}
\hspace{1cm}\includegraphics[scale=1.1]{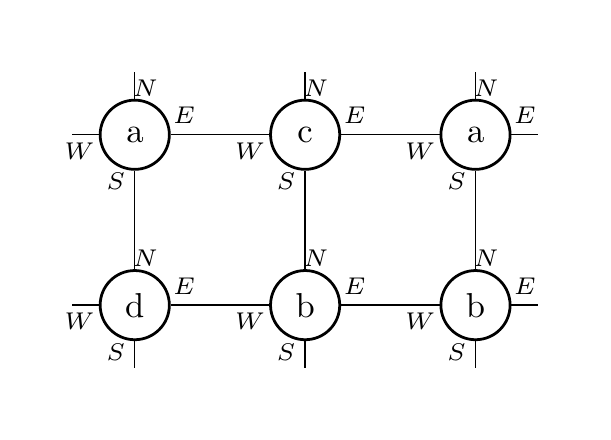}

\vspace{-0.5cm}
\includegraphics[scale=1]{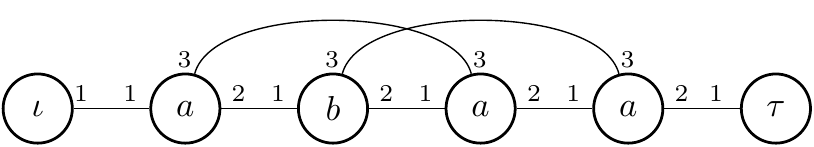}
\end{center}

\caption{(left) Rooted circular string, (right) graph associated to the 2D-word $aca\atop dbb$  and (bottom) graph associated to the string $abaa$ with extra edges connecting $u_n$ to $u_{2+n}$.}
\label{fig_circular_string}
\end{figure}

%% file: conclusion.tex
The model we propose naturally generalizes recognizable series on strings and trees. It satisfies closure properties by sum and Hadamard product. We have analysed why finite support series on some families of hypergraphs are not recognizable, and we exhibit a sufficient condition on families of hypergraph for the recognizability of finite support series. 

Since many data over a variety of fields naturally present a graph structure (images, secondary structure of RNA in bioinformatics, dependency graphs in NLP, etc.), this computational model offers a broad range of applications.

% These results suggest that the notion of HWM naturally extends the notion of linear representation for strings and trees, and that the set of recognizable series could be a natural extension of rational series to hypergraphs. 

The next theoretical step will be to study how learning can be achieved within this framework, i.e. how the tensor components of the model $M$ can be recovered or estimated from samples of the form $(G_1, \widehat{r_M(G_1)}), \dots, (G_n, \widehat{r_M(G_n)})$. Preliminary results on circular strings indicate that this is a promising direction. General learning algorithms should rely on tensor decomposition techniques, which generalize the spectral methods used for learning rational series on strings and trees. We also plan to tackle algorithmic issues and to study how techniques and methods developed in the field of graphical models, such as message passing, variational methods, etc., could be adapted to the setting of HWMs. 

%%% Local Variables: 
%%% mode: latex
%%% TeX-master: "main"
%%% End: 